\documentclass[alpha-refs]{wiley-article}

\usepackage{siunitx}

\usepackage{amsmath,amssymb,rotating,epsfig,graphicx}

\usepackage{graphics}
\usepackage{url} 
\usepackage{multirow}
  \usepackage{amsfonts}

\usepackage{bbm}

\usepackage{natbib}
\usepackage{xcolor}
\definecolor{darkblue}{rgb}{.0, .0,.6}
\definecolor{darkgreen}{rgb}{.1, .6,.3}

\usepackage{lscape}
\usepackage{tikz-qtree}

\newcommand{\Appendix}{
  \def\thesection{Appendix~\Alph{section}}
  \def\thesubsection{A.\arabic{subsection}}
}

\newcommand{\sN}{\scriptscriptstyle}

\newcommand{\pkg}[1]{{\normalfont\fontseries{b}\selectfont #1}}
\let\proglang=\textsf

\title{Automated Selection of Post-Strata using a Model-Assisted Regression Tree Estimator}
 
\author[1]{Kelly S. McConville}
\author[2]{Daniell Toth}

\affil[1]{Department of Mathematics and Statistics, Swarthmore College, Swarthmore, PA, 19081, USA}
\affil[2]{Office of Survey Methods Research, Bureau of Labor Statistics, Washington, DC, 20212, USA}

\corraddress{Kelly S. McConville, Department of Mathematics and Statistics, Swarthmore College, Swarthmore, PA, 19081, USA}
\corremail{kmcconv1@swarthmore.edu}

\fundinginfo{Funded by the American Statistical Association, National Science Foundation and Bureau of Labor Statistics (ASA/NSF/BLS) Fellowship program during the period 2016–2018}

\runningauthor{McConville and Toth}

\usepackage{Sweave}
\begin{document}
\Sconcordance{concordance:MAT_2017_08_26_k.tex:MAT_2017_08_26_k.Rnw:%
1 92 1 1 0 215 1}
\Sconcordance{concordance:MAT_2017_08_26_k.tex:./Section_Sim_2017_08_26.Rnw:ofs 309:%
1 5 1 1 26 2 1 1 32 17 1 1 3 7 0 1 2 2 1 1 3 7 0 1 2 16 1 2 2 1 1 2 2 1 %
1 2 2 1 1 2 2 14 1}
\Sconcordance{concordance:MAT_2017_08_26_k.tex:MAT_2017_08_26_k.Rnw:ofs 397:%
310 197 1}

\maketitle

\begin{abstract}
Auxiliary information can increase the efficiency of survey estimators through an assisting model when the model captures some of the relationship between the
auxiliary data and the study variables.
Despite their superior properties, model-assisted estimators are rarely used in anything but their simplest form by statistical agencies to produce official statistics.
This is due to the fact that the more complicated models that have been used in model-assisted estimation are
often ill suited to the available auxiliary data.
Under a model-assisted framework, we propose a regression tree estimator for a finite population total.  
Regression tree models are adept at handling the type of auxiliary data usually available in the sampling frame and provide a model that is easy to explain and justify.
The estimator can be viewed as a post-stratification estimator where the post-strata are automatically selected by the recursive partitioning algorithm of the regression tree.
We establish consistency of the regression tree estimator and compare its performance to other survey estimators using the US Bureau of Labor Statistics Occupational Employment Statistics Survey.

\end{abstract}

 \section{Introduction}

Much of the data used to make informed decisions come from surveys using probability sample designs to select units for data collection.
Indeed, governmental statistical agencies use sample data to produce estimates of official statistics consisting of population means and/or totals that are vital to inform policy makers and the public about the current state of the economy, public health, and environment among others topics.
The quality of these estimates varies between samples and depends on the amount and quality of data provided by the respondents.
This variability may be reduced and the estimator made more efficient if variables in an auxiliary dataset, associated with the study variable, are available for all units in the population.

Consider a finite population total $t_y = \sum_{j \in U_{\sN{}}} y_j$ where $y_j$ represents the value of the study variable for the $j$th unit in the population indexed by the set $U_{\sN{}} = \{1, 2, \ldots, N\}$.  Assume $d$ auxiliary variables, $\bm{x}_j = (x_{j1}, x_{j2}, \ldots, x_{jd})^{\sN{T}},$ are known for each $j \in U_{\sN{}}$.  After a random sample $s_{\sN{}}\subset U_{\sN{}}$ is selected using a complex sample design, with inclusion probabilities $\pi_j = P(j \in s \; | \; \bm{x}_j)$, the total $t_y$ can be estimated using the data from the sampled units.  
A standard approach is to use the Horvitz-Thompson (HT) estimator
\begin{align}\label{ht}
\hat{t}_{y,ht} = \sum_{j \in s_{\sN{}}}{w_j y_j}
\end{align}
\citep{hor52},
where $w_j = \pi_j^{-1}.$
The HT estimator is a design unbiased estimator of $t_y$. 
Under a purely design-based framework, statistical properties, such as bias and variance, are measured with respect to repeated sampling of units from the population and auxiliary data do not impact the form of the estimator but can impact $w_j$ through construction of the sample design.

%

Auxiliary information can be incorporated more directly into the estimation procedure through a model-assisted framework.  In this framework, the estimators are constructed to be asymptotically design unbiased and consistent regardless of model misspecification and are typically more efficient than estimators based solely on the sample data.  
Two common model-assisted approaches are calibration estimators \citep{dev92} and generalized regression estimators \citep{sar92}.  The calibration estimator is a weighted sum of the sampled study variable, similar to equation (\ref{ht}), but the weights, $\{\tilde{w}_j\}_{j \in s}$, are constructed so that they are close to the HT weights, $\{w_j\}_{j \in s}$, and reproduce (i.e. are calibrated to) known totals when applied to auxiliary variables
\begin{align*}
\sum_{j \in s}{\tilde{w}_j \bm{x}_j} = \sum_{j \in U} \bm{x}_j.
\end{align*}

In contrast, the generalized regression estimator employs a model to predict $y_j$ for each $j \in U_{\sN{}}$. The predictions are obtained by assuming the finite population values, $\{y_j, \bm{x}_j\}_{j \in U_{\sN{}}}$, are independent realizations from a superpopulation model, denoted by $\xi$, with mean function $h(\bm{x}_j) = \mbox{E}_{\xi}[Y_j | \textbf{X}_j=\bm{x}_j]$ where $\mbox{E}_{\xi}$ represents the expectation evaluated with respect to $\xi.$  Then, $h(\bm{x})$ is estimated by on the sampled observations.  Denoting an estimator of $h(\bm{x})$ as $\hat{h}_{\sN{n}}(\bm{x})$, the generalized regression estimator for $t_y$ is
\begin{eqnarray}\label{greg}
\widehat{t}_{y, r} = \sum_{j \in s} \frac{y_j - \hat{h}_{n}(\bm{x}_j)}{\pi_{\sN{}j}} + \sum_{j \in U} \hat{h}_{n}(\bm{x}_j).
\end{eqnarray}
The generalized regression estimator under a linear mean function can also be written as a weighted sum of the study variable, $y$:
\begin{eqnarray}\label{reg}
\widehat{t}_{y,l}& =& \sum_{j \in s} \left[1 + (\bm{t}_x - \widehat{\bm{t}}_{x,ht})^{\sN{T}} \left( \sum_{k \in s} \frac{\bm{x}_k \bm{x}_k^{\sN{T}}}{\pi_k} \right)^{-1} \bm{x}_j \right] \frac{1}{\pi_j} y_j
=\sum_{j\in s}\tilde{w}_j y_j
\end{eqnarray}
\citep{sar92}.
where $\bm{t}_x = \sum_{j \in U} \bm{x}_j$ and $\widehat{\bm{t}}_{x,ht}$ are the HT estimators of $\bm{t}_x$.  The sampling weights given by equation (\ref{reg}), $\{\tilde{w}_j\}_{j \in s}$, are calibrated to $\bm{t}_x$, and therefore, under the linear model, the generalized regression estimator is a calibration estimator.

The auxiliary data available for the target population of a survey often include categorical variables, often with a large number of categories.
Complex interactions can exist between variables and the variables are usually not independent. 
This makes it challenging to model the study variable using standard parametric modeling approaches.  For example, the Quarterly Census of Employment and Wages (QCEW) is the sampling frame for many establishment surveys at the US Bureau of Labor Statistics (BLS).
QCEW is compiled from administrative records from the state unemployment insurance programs and contains several categorical variables, such as industry code, whether or not the establishment is part of a multi-establishment firm, ownership-type, the state and region in which the establishment resides, and a size class for the establishment.  

A variety of linear models can be constructed from these categorical frame variables, ranging from an additive model to a model that includes all possible cross classifications. The generalized regression estimator based on an additive, linear model is equivalent to the generalized raking estimator \citep{dev93} and the saturated, linear model results in the post-stratification estimator \citep{smi91}. While the additive, linear model would produce an estimator that is calibrated on each individual category of the auxiliary variables, it would miss potentially informative interactions.  At the other extreme, calibrating on all possible interactions could produce significantly variable weights (sometimes negative weights), and could greatly increase the variance of the estimator. The weights produced by the estimator based on the additive, linear model may also be variable if one of the categorical variables has many categories. We will see that using a large number of variables with a linear model often leads to a poorly fit model, producing an estimator that is less efficient than a purely design-based estimator, such as the HT estimator.

Motivated by the ubiquity of categorical variables and interactions between variables in survey data, \cite{mor63} developed regression trees, which are binary trees based on a recursive partitioning algorithm, as a method for analyzing survey data. Instead of including all categories and/or all interactions, regression tree models group the categories of a variable based on how they relate with the study variable and only include variables and interactions associated with the study variable.



Recursive partitioning algorithms sequentially partition the data into two groups based on an auxiliary variable and estimate the mean in each group.  
The groups are selected by finding the split that provides the greatest reduction in the mean square error at each split.
This recursive splitting allows for complex interactive effects to easily be incorporated into the model. 
For instance, when splitting on a categorical variable, the categories of the variable are divided into the two groups that are most homogeneous with respect to the study variable.
The algorithm will only continue splitting if significant reduction in the mean square error is achieved.  
Therefore, inclusion of a categorical variable does not require a split for each category unlike the linear regression model.
This can substantially reduce the model size while still capturing interactions between categorical variables.

Regression trees have been applied to survey data in a variety of contexts such as analyzing nonresponse \citep{phi12}, nonresponse adjustment, \citep{lohr15}, hot deck imputation \citep{bec12}, and exploratory analysis \citep{ell15}.
We consider finite population estimation and study the behavior of the generalized regression estimator with a regression tree working model.  Other nonparametric models, such as local polynomial regression (\cite{bre00}; \cite{mon05}),  penalized splines (\cite{bre05}; \cite{mcc13}), and neural networks (\cite{mon05}), have also been explored for model-assisted estimation.  Generally, the model-assisted estimators based on a nonparametric model tend to be more efficient than their parametric counterparts when the relationship between the study variable and auxiliary variables is nonlinear.  A survey of these techniques can be found in \cite{bre17}. Despite the development of these more sophisticated model-assisted estimators, they are rarely ever used by statistical agencies to produce official statistics in anything but the simplest calibration framework where a linear model is assumed.  Since a regression tree can be framed as a linear model but overcomes the modeling issues often associated with the auxiliary data available to statistical agencies, we believe it has great potential as an estimation tool for official statistics.


Section \ref{sec:estimator} contains the development of the model assisted estimator using a design consistent regression tree model and shows how the estimator can be framed as a post-stratification estimator where the post-strata are given by the regression tree. This section also contains the statement of the main result that the model-assisted estimator is asymptotically design unbiased and consistent under general conditions on the sample design and population distribution.   Using the BLS Occupational Employment Statistics Survey data, three estimators: the regression tree estimator, the linear regression estimator, which is a generalized regression estimator with a linear working model, and the Horvitz-Thompson estimator, are compared in Section \ref{sec:sim}.  The regression tree estimator dominates the other estimators when much of the variability in the study variable can be explained by key interactions between the auxiliary variables. Section \ref{sec:conclusions} discusses implementation of the estimator and possible extensions.  
Conditions and proof of the main result are contained in the appendix.

\section{Proposed Estimator}\label{sec:estimator}

Many methods have been considered for estimating $h(\bm{x})$.  We propose estimating $h(\bm{x})$ based on a recursive partitioning algorithm applied to $\{y_j, \bm{x}_j\}_{j \in s_{\sN{}}}$. Let $Q_{n} = \{ B_{n1}, B_{n2}, \ldots, B_{n q} \}$ represent the partitioning boxes of the sample based on the algorithm and let $I(\bm{x}_j \in B_{nk}) = 1,$ if $\bm{x}_j \in B_{nk},$ and 0, otherwise.
For $\bm{x}_j \in B_{nk},$ the  estimator of $h(\bm{x}_j)$ is given by
 \begin{equation}\label{hn_tilde}
    \tilde{h}_n(\bm{x}_j)=
   \widetilde{\#}_{\sN{}}(B_{nk})^{-1} \sum_{j \in s_{\sN{}}} \pi_{\sN{}j}^{-1} y_j I (\bm{x}_j \in B_{n k})
\end{equation}
where 
$$\widetilde{\#}_{\sN{}}(B_{nk}) = \displaystyle\sum_{j \in s_{\sN{}}} \pi_{j}^{-1} I(\bm{x}_j \in B_{nk}),$$ 
the Horvitz-Thompson estimator of the population size in box $B_{nk}$.  Notice $\tilde{h}_n(\bm{x}_j)$ is the Hajek (\citeyear{haj71}) estimator of the population mean for box $B_{nk}$. \cite{tot11} show that $\tilde{h}_n(\bm{x})$ is consistent for the superpopulation mean function, $h(\bm{x})$. The necessary regularity conditions for the recursive partitioning algorithm so that the estimator is consistent are given in the Appendix.

The regression tree estimator, $t_{y,t}$ is obtained by inserting equation (\ref{hn_tilde}) into the generalized regression estimator, given in equation (\ref{greg}).
Since $\tilde{h}_n(\bm{x})$ is random with respect to the design, $\widehat{t}_{y,t}$ is not design unbiased for $t_y$, but is asymptotically design unbiased and consistent under general conditions on the sample design and the superpopulation model. The conditions and proofs to establish the following asymptotic result are presented in the Appendix.
\begin{theorem}\label{theo:consist}
Under assumptions \ref{cond1}- \ref{cond5}, 
$N^{-1} \widehat{t}_{y,t}$ is asymptotically design unbiased in the sense that
$$
\lim_{N\to\infty}\mbox{E}_p{N^{-1}\left(\widehat{t}_{y,t} - t_y\right) } = 0,
$$
and design consistent in the sense that
$$
\lim_{N\to\infty}\Pr\left\{ N^{-1}\left|\widehat{t}_{y,t} - t_y\right| >  \eta\right\} = 0
$$
for all $\eta > 0$.
\end{theorem}
The proof of the above result follows the method used in \cite{tot11}, but because of the dependence between the splits in a regression tree and the study variable, consistency of the model assisted estimator using a regression tree model does not follow directly from the consistency of the regression tree model shown in that paper.
As a consequence, the proofs require adjusted rates, including a different sampling fraction, to those given in \cite{tot11}.

\subsection{Post-Stratification}

From equation (\ref{hn_tilde}) we see that for a given partition, $Q_n$, the mean estimator, $\tilde{h}_n(\bm{x})$, can be written as a linear regression estimator with $q$ indicator function covariates
\begin{align*}
\tilde{h}_n(\bm{x}) = \tilde{\mu}_{n1}  I(\bm{x} \in B_{n1}) + \tilde{\mu}_{n2} I(\bm{x} \in B_{n2}) + \ldots + \tilde{\mu}_{nq}  I(\bm{x} \in B_{nq})
\end{align*}
where
\begin{align*}
\tilde{\mu}_{nk} = \left[\widetilde{\#}\left(B_{nk}\right) \right]^{-1} \sum_{j \in s} \pi_j^{-1} y_j I(\bm{x}_j \in B_{nk}).
\end{align*}
Therefore, the regression tree estimator is also a post-stratification estimator where each box, $B_{nk}$, represents a post-strata.  This also implies that the estimator is calibrated to the population total of each box. From this perspective, the tree provides a data-driven mechanism for selecting the post-stratification where none of the resulting post-strata are empty.


\section{Simulation Study using Occupational Employment Statistics Survey Data}\label{sec:sim}



The BLS Occupational Employment Statistics survey (OES) is a semi-annual establishment survey collecting occupational employment and wage data for workers in non-farm industries. 
The survey data are used to publish annual occupational employment and wage estimates for over 800 occupations at the metropolitan and micropolitan statistical area (MSA) level and for specific industries at the national level.
These estimates are used by the public and policy-makers to identify occupations that are growing or shrinking in the economy and to make labor projection forecasts.

Using a sample design that is stratified by area and industry, a sample of about 200,000 establishments are selected each May and November where establishments are selected with probability proportional to size within each stratum.
Sample weights of the selected sample units are adjusted so that the weighted employment totals are calibrated to QCEW employment totals. 
See \cite{bls08} for more details on the sample design of the OES.


The frame of establishments used by the OES is derived from the QCEW, and contains the size class of the MSA where the establishment is located, whether or not the establishment is part of a multi-establishment company, and the establishment's industry and size-class, for each establishment in the population.
Since the number of employees at an establishment with a specific occupation code varies between types of establishment, a model-assisted estimator which accounts for the establishment characteristics should be more efficient than an purely design-based estimator.  However, usually there are only a few establishment-level variables or categories related to occupational employment and how these variables relate to employment numbers depends on the occupation. Therefore, the model and auxiliary variables used should vary for each occupation employment estimator.

Consider the regression tree model shown in Figure \ref{tree:bar}, relating the number of bartenders per establishment to the establishment characteristic variables available in the QCEW data.
The first (and by far the most important) split identifies establishments with industry codes 71 (Arts, Entertainment, and Recreation) which include casinos, ski-resorts and amusement parks, and 72 (Accommodation and Food Services) which includes bars and restaurants, as having a higher average number of employees with the occupation of bartender than establishments in other industries.  This is not surprising for this occupation, but you would not expect the same split for other occupations.

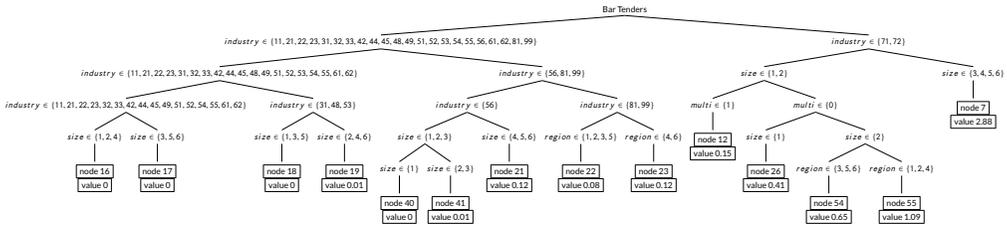
\begin{figure}[htb] 
\centering 
\begin{tikzpicture}[scale=0.4, ] 
\tikzset{every tree node/.style={align=center,anchor=north}} 
\Tree [.{Bar Tenders} [.{$industry \in \{11,21,22,23,31,32,33,42,44,45,48,49,51,52,53,54,55,56,61,62,81,99\}$} [.{$industry \in \{11,21,22,23,31,32,33,42,44,45,48,49,51,52,53,54,55,61,62\}$} [.{$industry \in \{11,21,22,23,32,33,42,44,45,49,51,52,54,55,61,62\}$} [.{$size \in \{1,2,4\}$} {\fbox{node 16} \\ \fbox{value 0}} ][.{$size \in \{3,5,6\}$} {\fbox{node 17} \\ \fbox{value 0}} ]][.{$industry \in \{31,48,53\}$} [.{$size \in \{1,3,5\}$} {\fbox{node 18} \\ \fbox{value 0}} ][.{$size \in \{2,4,6\}$} {\fbox{node 19} \\ \fbox{value 0.01}} ]]][.{$industry \in \{56,81,99\}$} [.{$industry \in \{56\}$} [.{$size \in \{1,2,3\}$} [.{$size \in \{1\}$} {\fbox{node 40} \\ \fbox{value 0}} ][.{$size \in \{2,3\}$} {\fbox{node 41} \\ \fbox{value 0.01}} ]][.{$size \in \{4,5,6\}$} {\fbox{node 21} \\ \fbox{value 0.12}} ]][.{$industry \in \{81,99\}$} [.{$region \in \{1,2,3,5\}$} {\fbox{node 22} \\ \fbox{value 0.08}} ][.{$region \in \{4,6\}$} {\fbox{node 23} \\ \fbox{value 0.12}} ]]]][.{$industry \in \{71,72\}$} [.{$size \in \{1,2\}$} [.{$multi \in \{1\}$} {\fbox{node 12} \\ \fbox{value 0.15}} ][.{$multi \in \{0\}$} [.{$size \in \{1\}$} {\fbox{node 26} \\ \fbox{value 0.41}} ][.{$size \in \{2\}$} [.{$region \in \{3,5,6\}$} {\fbox{node 54} \\ \fbox{value 0.65}} ][.{$region \in \{1,2,4\}$} {\fbox{node 55} \\ \fbox{value 1.09}} ]]]][.{$size \in \{3,4,5,6\}$} {\fbox{node 7} \\ \fbox{value 2.88}} ]]]\end{tikzpicture} 

\caption{\small Tree model on full OES data set for predicting the number of bartenders per establishment. The survey-weighted average number of bartenders is given at each end node.} 
\label{tree:bar} 
\end{figure} 
Figure \ref{tree:eled} shows the tree model for elementary school teachers. The first split separates industry code 61 (Educational Services), which includes schools, from all other industries.  Again, this is as expected, as is the fact that larger schools have a higher number of elementary school teachers than smaller schools.  

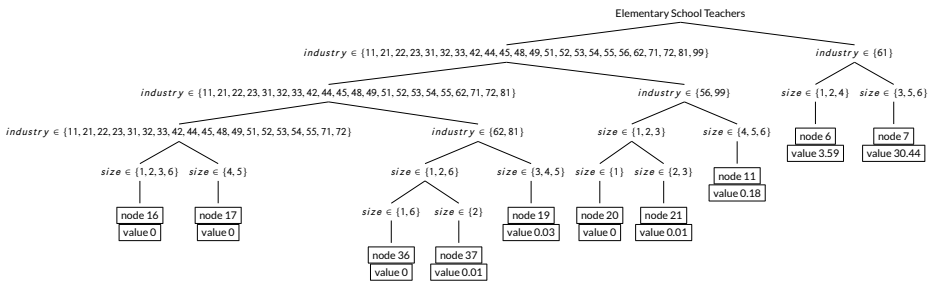
\begin{figure}[htb] 
\centering 
\begin{tikzpicture}[scale=0.5, ] 
\tikzset{every tree node/.style={align=center,anchor=north}} 
\Tree [.{Elementary School Teachers} [.{$industry \in \{11,21,22,23,31,32,33,42,44,45,48,49,51,52,53,54,55,56,62,71,72,81,99\}$} [.{$industry \in \{11,21,22,23,31,32,33,42,44,45,48,49,51,52,53,54,55,62,71,72,81\}$} [.{$industry \in \{11,21,22,23,31,32,33,42,44,45,48,49,51,52,53,54,55,71,72\}$} [.{$size \in \{1,2,3,6\}$} {\fbox{node 16} \\ \fbox{value 0}} ][.{$size \in \{4,5\}$} {\fbox{node 17} \\ \fbox{value 0}} ]][.{$industry \in \{62,81\}$} [.{$size \in \{1,2,6\}$} [.{$size \in \{1,6\}$} {\fbox{node 36} \\ \fbox{value 0}} ][.{$size \in \{2\}$} {\fbox{node 37} \\ \fbox{value 0.01}} ]][.{$size \in \{3,4,5\}$} {\fbox{node 19} \\ \fbox{value 0.03}} ]]][.{$industry \in \{56,99\}$} [.{$size \in \{1,2,3\}$} [.{$size \in \{1\}$} {\fbox{node 20} \\ \fbox{value 0}} ][.{$size \in \{2,3\}$} {\fbox{node 21} \\ \fbox{value 0.01}} ]][.{$size \in \{4,5,6\}$} {\fbox{node 11} \\ \fbox{value 0.18}} ]]][.{$industry \in \{61\}$} [.{$size \in \{1,2,4\}$} {\fbox{node 6} \\ \fbox{value 3.59}} ][.{$size \in \{3,5,6\}$} {\fbox{node 7} \\ \fbox{value 30.44}} ]]]\end{tikzpicture} 

\caption{\small Tree model on full OES data set for predicting the number of elementary school teachers (excluding special education) per establishment. The survey-weighted average number of elementary school teachers is given at each end node.} 
\label{tree:eled} 
\end{figure} 
While bartenders and elementary school teachers are predominately found in different industry sectors, they are both service occupations, so they are both found, to a much lesser extent, in some of the same service sector industry codes.
Both are found, for example, in industry code 56 (Administrative and Support Services) and industry code 81 (Other Services). Industry code 56 includes convention centers which employ bartenders and also includes correctional and support facilities which  occasionally  employ elementary school teachers, while industry code 81 includes social clubs with bartenders and daycare centers which sometimes have elementary school teachers.
Most of the following splits are on size of the establishments, where again (not surprisingly) establishments with more employees are expected to have more employees with the respective occupational code than establishments with fewer employees in the same industry group. 

Treating the OES data as the target population, we conducted a simulation study to compare the performance of the regression tree estimator to a linear regression estimator and to a Horvitz-Thompson (HT) estimator.  For each estimator, we estimated employment totals for a variety of occupations. 
The test was done over 1000 repeated samples of size $n=\{2000, 3000, 4000, 5000, 6000\}$ from the 187,115 establishments in the OES dataset using a stratified design where establishments in larger size classes were sampled with higher probability than establishments in smaller size classes, resulting in an unequal probability sample design.

Since the regression tree model conducts variable selection at each split, it provides an automated method for choosing which variable and category combinations form the post-strata on which the estimator is calibrated.  For a comparable contender, we also automated the calibration selection for the linear model by using a forward step-wise procedure to select variables for inclusion in the model.  However, since the two variable selection procedures differ, the resulting calibrations also differ. The stepwise linear estimator is calibrated to all categories of the variables selected by the procedure while the regression tree estimator is calibrated to the sub-groups represented by the tree's resulting partition.  For example, an regression tree estimator based on the tree in Figure \ref{tree:eled} would be calibrated to the total number of establishments in each of the 10 sub-groups defined by each end-node.
A linear regression estimator with the same two variables (\textit{industry, size}) would be calibrated to the totals of each of the 30 categories, but no interactions between categories.  We also considered a stepwise linear estimator based on all possible interactions but the estimator performed very poorly and therefore is not included in the results.

For each sample, the models were fit to the sample data using weights equal to the inverse of their inclusion probabilities. We used the forward-stepwise variable selection procedure in \proglang{R} \citep{R} for the linear regression models and the \proglang{R} package \pkg{rpms} \citep{rpms} to build the regression trees.  An estimate of the total employment for each occupation tested was produced for each estimator.

\begin{figure}[htb]
\begin{tabular}{c}
\vspace{-.8in}
\includegraphics{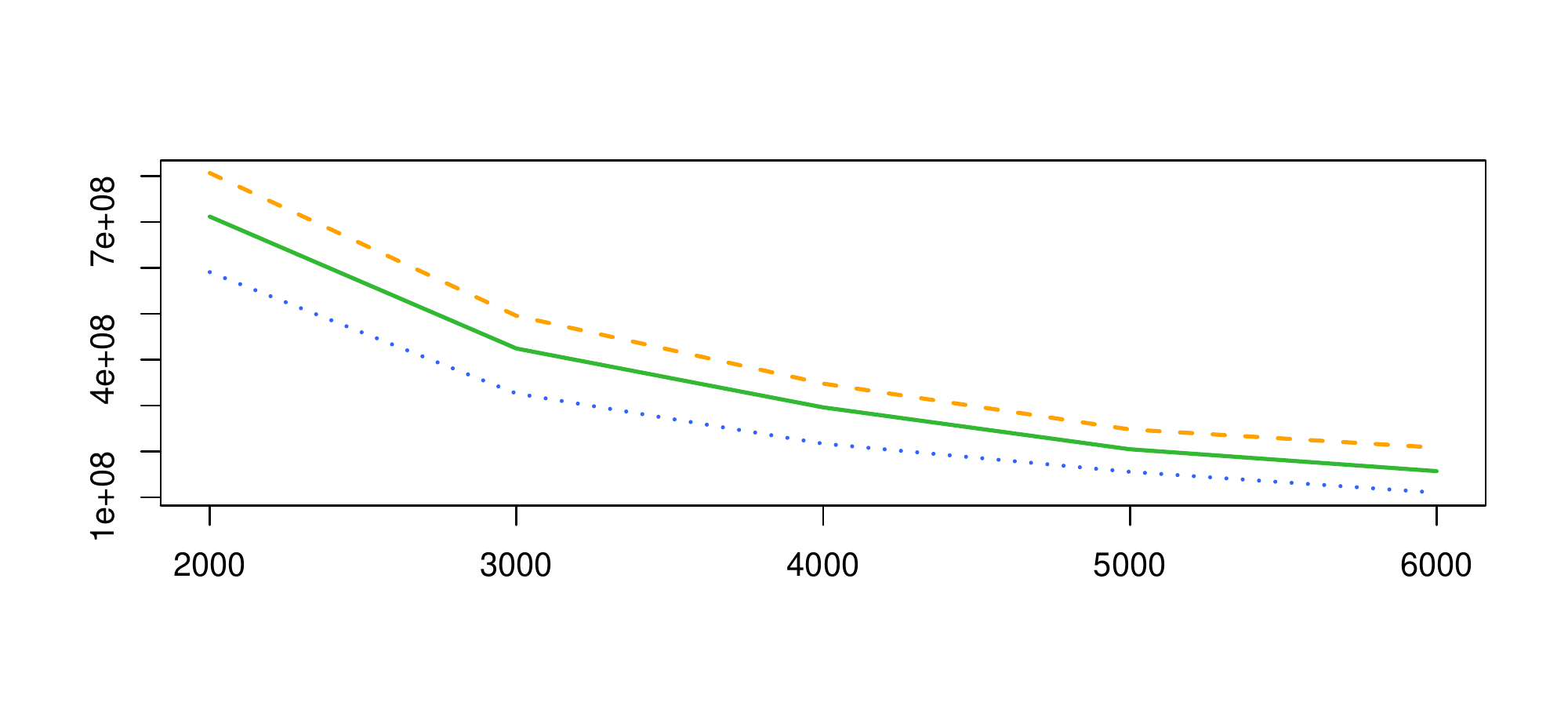}
\\
\vspace{-.8in}
\includegraphics{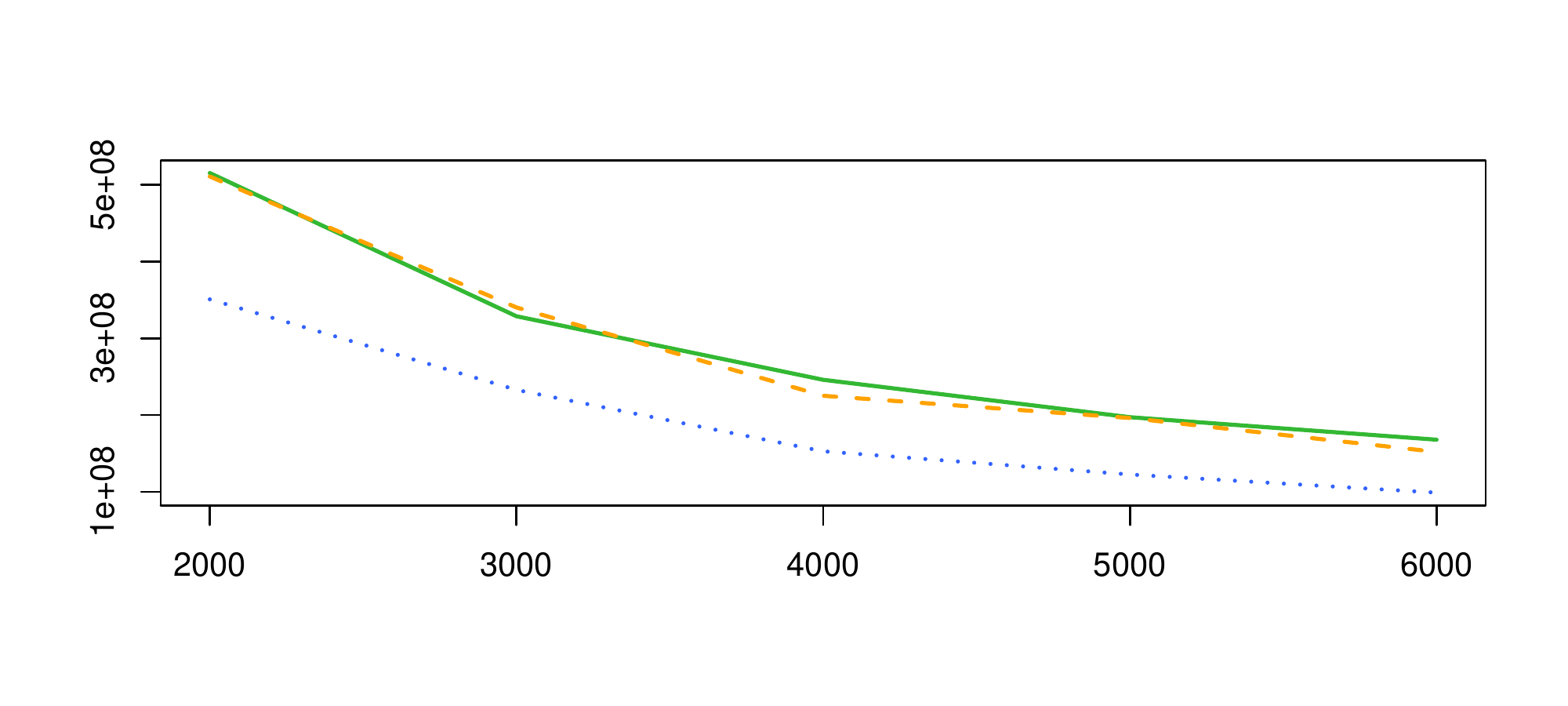}
\\
\vspace{-.8in}
\includegraphics{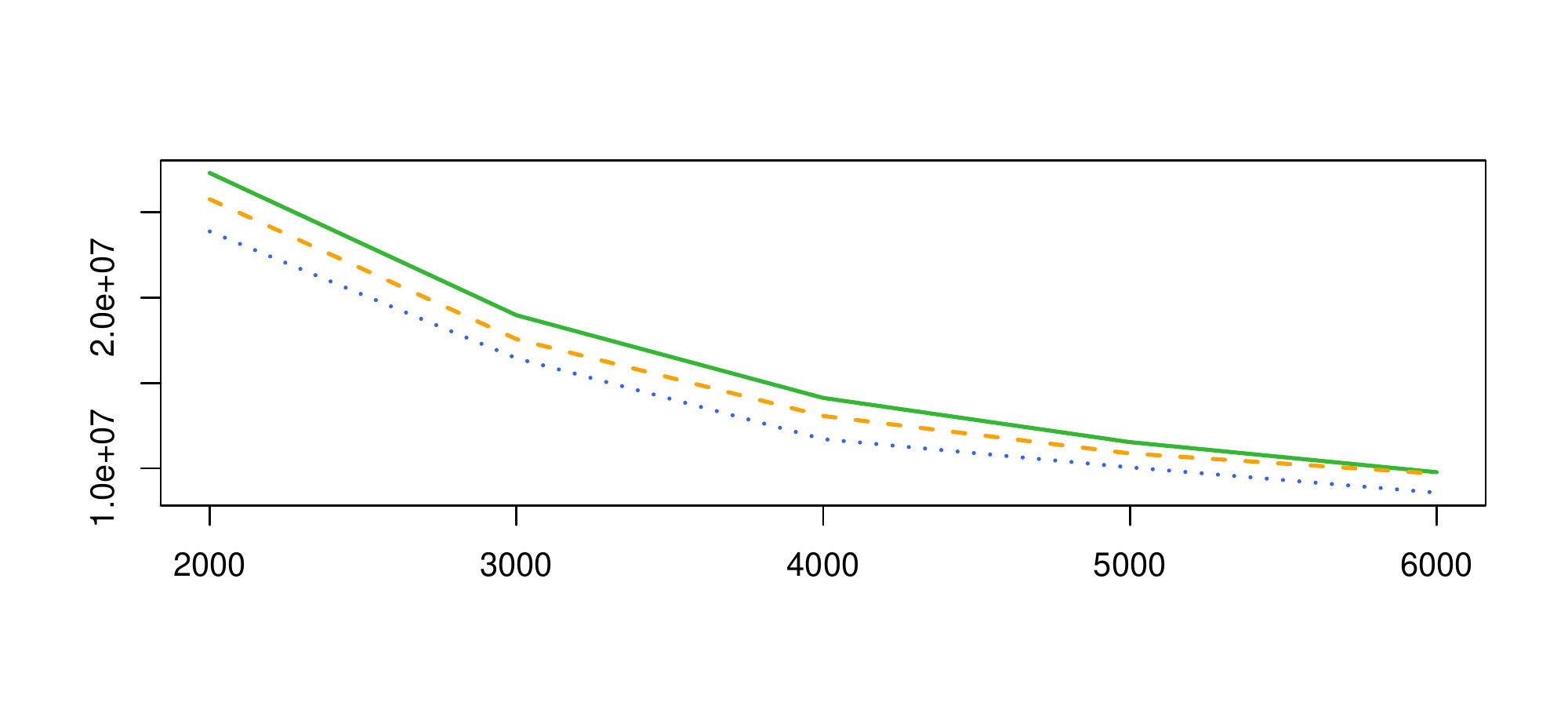}
\\
\vspace{-.3in}
\includegraphics{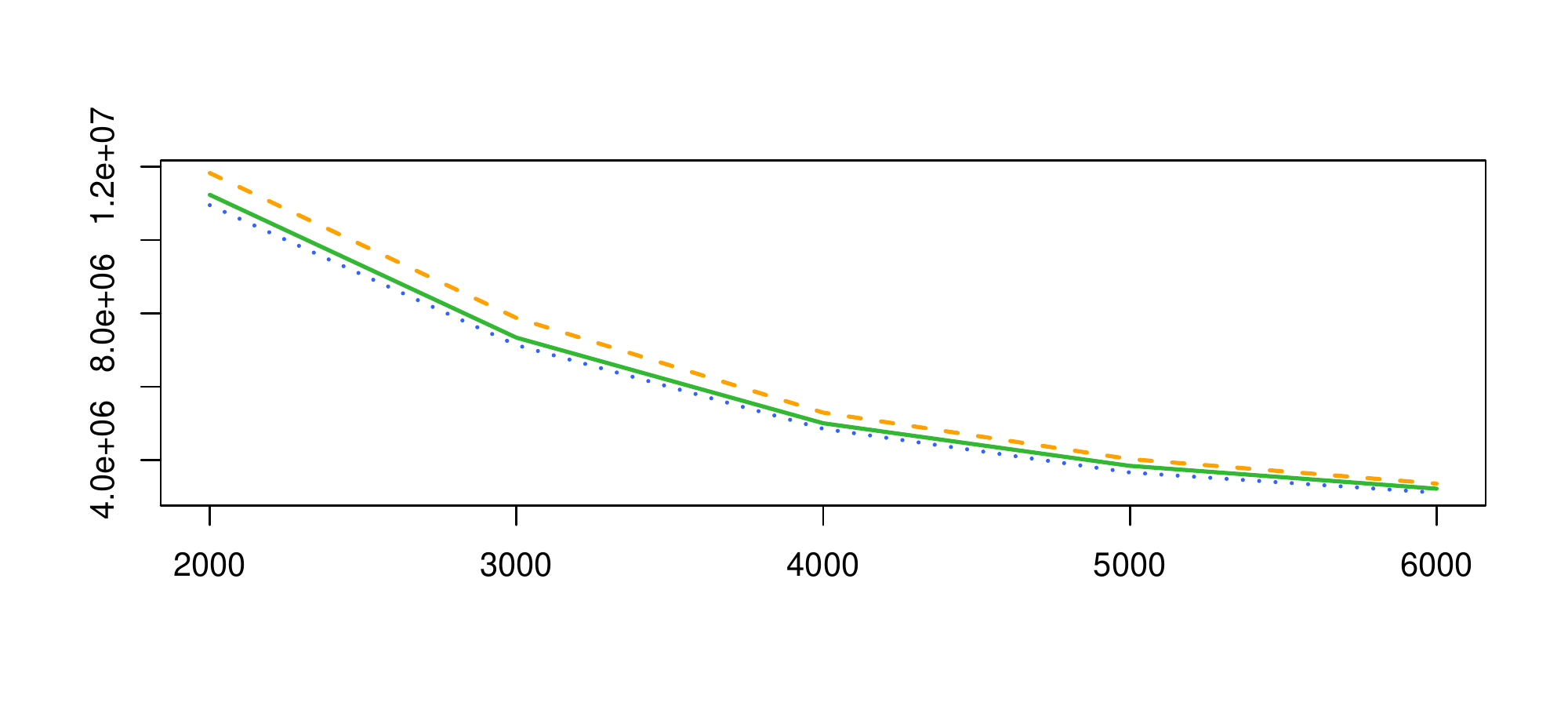}
\\
\end{tabular}
\caption{Simulation results comparing the Horvitz-Thompson estimator (green solid-line), linear regression estimator (orange dashed-line) and the regression tree estimator (blue dotted-line) of employment total for the occupations: (top to bottom) elementary school teachers, waiters and waitresses, bartenders, and sales-managers.  The graphs display the mean squared error of the estimates for the 1,000 repeated samples by sample size.}\label{fig:mse}
\end{figure}

The empirical mean squared error (MSE) of the 1000 total estimates by sample size for each of the three estimators is shown in Figure \ref{fig:mse} for four of the occupations that we tested: elementary school teachers, waiters and waitresses, bartenders, and sales-managers. The top two graphs in Figure \ref{fig:mse} contain the test results for elementary school teachers and for waitstaff, respectively. Both graphs show a large increase in efficiency using the regression tree estimator (dotted-line) compared to using the linear regression estimator (dashed-line) and the HT estimator (solid line).
This is in contrast to the results for bartender and sales-manager occupation codes, which are shown in the bottom two graphs in Figure \ref{fig:mse}, respectively.
These graphs show no real gain in MSE using either of the model-assisted estimators compared to the HT estimator.
It is worth noting that the top two occupation codes are large occupation classes with a total employment in our OES finite population of 205,352 for elementary school teachers and 213,206 for waitstaff, while the occupation codes bartender and sales-manager have much smaller total employment of 35,811 and 34,585 respectively.

While for all four occupation codes the regression tree estimator is either more efficient or has about the same efficiency as the HT estimator, the linear regression estimator is never more efficient than the HT estimate and for elementary school teachers and sales-managers, it is less efficient.
The results for these four occupation codes are consistent with the results obtained for all the occupation codes tested, in that, the regression tree estimator was either more efficient or of similar efficiency as the HT estimator with efficiency gains occurring in larger occupations, while the linear regression estimator had either similar efficiency as the HT estimator or was occasionally much less efficient.


\section{Conclusions}\label{sec:conclusions}

In this article we have presented the regression tree estimator for a finite population total and have developed the design consistency of the estimator under standard assumptions on the sample design and finite population.  Following a model-assisted framework, the estimator is a generalized regression estimator where the working model is a regression tree. An implementation of the estimator will be available in an \proglang{R} package 
on The Comprehensive R Archive Network (CRAN).

The regression tree estimator is also a post-stratification estimator since a regression tree can be written as a linear model where each variable is an indicator function for the sequential splits to an end node of the tree. Therefore, the regression tree can be viewed as a data-driven technique for choosing the appropriate set of post-strata.  These post-strata have the ability to capture complex interactions between the variables and as shown in the simulations, they can increase the efficiency of the model-assisted estimator.  Additionally, the estimator is calibrated to the population totals of each post-strata. 

Since the auxiliary data available for establishment surveys are often mostly categorical, we have focused on useful features of how regression trees handle categorical data, such as, collapsing categories into homogeneous sub-groups and capturing predictive interactions between specific categories. However, regression trees are also useful for quantitative auxiliary data, especially if the data have a non-linear relationship with the study variable.  

There are a couple of extensions of the regression tree estimator to be considered in future work. If the study variable relates linearly with an auxiliary variable, then a simple linear regression model instead of a simple mean model could be fit to each of the end nodes.  Another option would be to use a random forest \citep{bre01}, which is a collection of regression trees fit on subsets of the variables and/or data, in the generalized regression estimator since a random forest tends to have greater predictive power than a single regression tree.


\section*{Appendix: Assumptions and Proof of Theorem \ref{theo:consist}}\label{sec:results}
\Appendix
\setcounter{section}{1}

\textbf{Rate conditions}: Define two functions of $N_{}$ that will be used as rates of convergence.
Let $\gamma(N_{ })$ and $k(N_{ })$ be given functions
bounded above 0 for all $N_{ }>0$ satisfying:

\renewcommand{\theenumi}{R\arabic{enumi}}
\begin{enumerate}
 \item : $\gamma(N_{ })  \rightarrow  \infty $ \label{rate1}
  \item : $N_{ }^{-1}k(N_{ })  \rightarrow  0$ \label{rate2}
  \item : $k(N_{ })^{-1}\gamma(N_{ })N^{1/2}_{ } \rightarrow  0.$ \label{rate3}
\end{enumerate}


Assume we have a sequence of nested populations, $U_1 \subset U_2 \subset \cdots \subset U_{\sN{}} \subset \cdots$. Let the sample, $s_{\sN{N}}\subset U_{\sN{N}}$, be selected according to a sample design $\rm{p}_{\sN{N}}(\cdot)$ with sample size $n_{\sN{N}}$. Denote the  first-order inclusion probabilities by $\pi_{\sN{N} j}=\Pr\{j,\in s_{\sN{N}}\}=\sum_{s_{\sN{N}} \subset U_{\sN{N}} :j\in s_{\sN{N}}}\rm{p}_{\sN{N}}(s_{\sN{N}})$ and the second-order inclusion probabilities by $\pi_{\sN{N} jk}=\Pr\{j,k\in s_{\sN{N}}\}=\sum_{s_{\sN{N}} \subset U_{\sN{N}} :j,k\in s_{\sN{N}}}\rm{p}_{\sN{N}}(s_{\sN{N}})$.  Let $I_{\sN{N}j}$ be an indicator function representing whether unit $j$ is or is not in the sample.  Let $\mbox{E}_p$ and $\mbox{Var}_p$ denote the expectation and the variance evaluated with respect to the sample design, $\rm{p}_{\sN{}}(\cdot)$ We suppress the subscript $N$ in $n$ as well as in, $I_j$, $\pi_j$ and $\pi_{jk}$  for simplicity of notation. Assume:

\renewcommand{\theenumi}{A\arabic{enumi}}
\begin{enumerate}

\item  For any set $A \in U_{N}$, $\left( \sum_{i \in U_N}I(i \in A)\right) ^{-1} \sum_{i \in U_N}
  y_i^2 I(i \in A) \leq M < \infty$ with $\xi$-probability 1. \label{cond1}


      \item $ \limsup_{ N \rightarrow \infty}
	   \big(N_{ } \min_{i \in U_{N}} \pi_{  i}\big)^{-1} = O(n^{-1}_{ })
	        \hspace{.3 cm} \mbox{with $\xi$-probability 1.}$  \label{pimin}


  \item $\limsup_{ N \rightarrow \infty} \max_{i, j \in U_{N} :i \neq j}
   \left| \pi_{  ij}\pi_{  i}^{-1}\pi_{  j}^{-1} - 1 \right| = O\big(n_{ }^{-1}\big)
        \hspace{.3 cm} \mbox{with $\xi$-probability 1.}$ \label{cond3}

    \item For $i \in U_N$, $\mbox{E}_p\Big[I_{  i} \;|\;Q_{n_{ }} \Big] =
          \pi_{  i } + d_i$ where $\mbox{E}_p \left(\max_{i \in U_N} |d_i | \right) = O(k(n)^{-1})$ and for $i, j \in U_{\sN{N}}: i \neq j$, $\mbox{E}_p\Big[I_{  i} I_{  j} \;|\;Q_{n_{ }} \Big] =   \pi_{  i j} + d_{ij}$ where $\mbox{E}_p \left( \max_{i, j \in U_{N} :i \neq j} | d_{ij}| \right)= O( \gamma(n)^2 k(n)^{-2}).$ \label{cond4}

%
%

\item The sampling fraction is given by $f = n N^{-1}.$  Let $f^{-1} = O(\gamma(n)^{-1} n^{1/2})$ and $f = O(\gamma(n) n^{-1/2})$.\label{cond5}
\end{enumerate}

%

\begin{proof}[Proof of Theorem \ref{theo:consist}]

By Markov's Inequality, we can achieve asymptotic design unbiasedness and design consistency by showing that
\begin{align*}
\underset{N \rightarrow \infty}{\lim} \text{ } \mbox{E}_p \left|\frac{\widehat{t}_{y,t} - t_y}{N} \right| = 0.
\end{align*}

Without loss of generality, assume the support of $\bm{X}_j$ is $(0,1)^{d}$. 
First, we define a population mean estimator of $h(\bm{x})$ using the sample-based partition, $Q_n$,
\begin{align*}
\tilde{h}_{\sN{N}}(\bm{x}) = \tilde{\mu}_{\sN{N}1}  I(\bm{x} \in B_{n1}) + \tilde{\mu}_{\sN{N}2} I(\bm{x} \in B_{n2}) + \ldots + \tilde{\mu}_{\sN{N}q}  I(\bm{x} \in B_{nq})
\end{align*}
where 
\begin{align*}
\tilde{\mu}_{\sN{N}k} = \left[{\#}\left(B_{nk}\right) \right]^{-1} \sum_{i \in U_{\sN{N}}} y_i I(\bm{x} \in B_{nk}).
\end{align*}

Proposition 1 of \cite{tot11} prove the mean square error consistency of the sample estimator, $\tilde{h}_n(\bm{x}),$  to the superpopulation mean, $h(\bm{x})$.  Within the proof of the result, mean square error consistency of the sample estimator to the finite population estimator using the sample-based partition is provided.  In particular, under assumptions  \ref{cond1}  -\ref{cond5}, Proposition 1 of \cite{tot11} implies that $\mbox{E}_p\left[( \tilde{\mu}_{\sN{N}k} - \tilde{\mu}_{nk} )^2 \right] = \gamma(n) n^{1/2}k(n)^{-1}$.


\noindent By adding and subtracting $$ \sum_{i\in U_{\sN{N}}} \frac{\tilde{h}_{\sN{N}}(\bm{x}_i) }{N}\left(\frac{I_i}{\pi_i}-1\right)$$ and applying the Triangle Inequality, we get
\begin{align*}
\mbox{E}_p \left|\frac{\widehat{t}_{y,t} - t_y}{N} \right| &=\mbox{E}_p \left|\sum_{i\in U_{\sN{N}}} \frac{\left(y_i- \tilde{h}_{\sN{N}}(\bm{x}_i)\right)}{N}\left(\frac{I_i}{\pi_i}-1\right)
+ \sum_{i\in U_{\sN{N}}} \frac{\left( \tilde{h}_{\sN{N}}(\bm{x}_i) - \tilde{h}_n(\bm{x}_i)\right)}{N}\left(\frac{I_i}{\pi_i}-1\right) \right|\\
&\leq \mbox{E}_p \left|\sum_{i\in U_{\sN{N}}} \frac{ \left(y_i- \tilde{h}_{\sN{N}}(\bm{x}_i)\right) }{N}\left(\frac{I_i}{\pi_i}-1\right) \right|
+ \mbox{E}_p \left| \sum_{i\in U_{\sN{N}}} \frac{\left( \tilde{h}_{\sN{N}}(\bm{x}_i) -
  \tilde{h}_n(\bm{x}_i)\right)}{N}\left(\frac{I_i}{\pi_i}-1\right)
  \right|\\
& := A_{\sN{N}} + B_{\sN{N}}.
\end{align*}

\noindent For the first term, we have
\begin{align*}
A_{\sN{N}} & \leq	\mbox{E}_p\left[ \underset{i, j \in U_{\sN{N}}}{\sum \sum} \frac{ \left(y_i- \tilde{h}_{\sN{N}}(\bm{x}_i)\right) }{N}\frac{ \left(y_j- \tilde{h}_{\sN{N}}(\bm{x}_j)\right) }{N}\left(\frac{I_i}{\pi_i}-1\right) \left(\frac{I_j}{\pi_j}-1\right) \right] \\
& = \mbox{E}_p\left[ \underset{i \in U_{\sN{N}}}{\sum } \frac{ \left(y_i- \tilde{h}_{\sN{N}}(\bm{x}_i)\right)^2 }{N^2}\left(\frac{I_i}{\pi_i}-1\right)^2 \right] \\
&+ \mbox{E}_p\left[ \underset{i\neq j: i, j \in U_{\sN{N}}}{\sum \sum} \frac{ \left(y_i- \tilde{h}_{\sN{N}}(\bm{x}_i)\right) }{N}\frac{ \left(y_j- \tilde{h}_{\sN{N}}(\bm{x}_j)\right) }{N}\left(\frac{I_i}{\pi_i}-1\right) \left(\frac{I_j}{\pi_j}-1\right) \right] \\
& := A_{\sN{N}1} + A_{\sN{N}2}.
\end{align*}
Applying the law of total expectation to the first term gives
\begin{align*}
A_{\sN{N}1} &= \mbox{E}_p \left\{ \mbox{E}_p\left[ \underset{i \in U_{\sN{N}}}{\sum } \frac{ \left(y_i- \tilde{h}_{\sN{N}}(\bm{x}_i)\right)^2 }{N^2}\left(\frac{I_i}{\pi_i}-1\right)^2 \Big|Q_n \right] \right\}\\
&= \mbox{E}_p \left\{  \underset{i \in U_{\sN{N}}}{\sum } \frac{ \left(y_i- \tilde{h}_{\sN{N}}(\bm{x}_i)\right)^2 }{N^2}\mbox{E}_p\left[\left(\frac{I_i}{\pi_i}-1\right)^2 \Big|Q_n \right] \right\}\\
&= \mbox{E}_p \left\{  \underset{i \in U_{\sN{N}}}{\sum } \frac{ \left(y_i- \tilde{h}_{\sN{N}}(\bm{x}_i)\right)^2 }{N^2} \left[\frac{(1-\pi_i)}{\pi_i} + \frac{d_i(1-2\pi_i)}{\pi_i^2} \right] \right\}\\
&\leq \mbox{E}_p \left\{  \underset{i \in U_{\sN{N}}}{\sum } \frac{ \left(y_i- \tilde{h}_{\sN{N}}(\bm{x}_i)\right)^2 }{N^2} \left|\frac{(1-\pi_i)}{\pi_i} \right| \right\}  +\mbox{E}_p \left\{  \underset{i \in U_{\sN{N}}}{\sum } \frac{ \left(y_i- \tilde{h}_{\sN{N}}(\bm{x}_i)\right)^2 }{N^2} \left| \frac{d_i(1-2\pi_i)}{\pi_i^2} \right| \right\}\\
&\leq \frac{1}{N \min_{i \in U_N} \pi_i} 4M^2 +  \mbox{E}_p \left(\max_{i \in U_N} |d_i|\right) \frac{1}{ N(\min_{i \in U_N} \pi_i)^2} 4M^2 \\
&= O(n^{-1}) +O(n^{-1}) O( n_{ }^{1/2}\gamma(n_{ })^{-1}k(n_{ })^{-1})\\
& =  O(n^{-1}) + o(n^{-1})
\end{align*}
by conditions (\ref{cond1}), (\ref{pimin}),  (\ref{cond4}) and (\ref{cond5}) and rate condition (\ref{rate3}). Similarly, the second term becomes
\begin{align*}
A_{\sN{N}2} & =  \mbox{E}_p\left\{ \mbox{E}_p \left[ \underset{i\neq j: i, j \in U_{\sN{N}}}{\sum \sum} \frac{ \left(y_i- \tilde{h}_{\sN{N}}(\bm{x}_i)\right) }{N}\frac{ \left(y_j- \tilde{h}_{\sN{N}}(\bm{x}_j)\right) }{N}\left(\frac{I_i}{\pi_i}-1\right) \left(\frac{I_j}{\pi_j}-1\right) \Big| Q_n \right] \right\} \\
& =  \mbox{E}_p\left\{  \underset{i\neq j: i, j \in U_{\sN{N}}}{\sum \sum} \frac{ \left(y_i- \tilde{h}_{\sN{N}}(\bm{x}_i)\right) }{N}\frac{ \left(y_j- \tilde{h}_{\sN{N}}(\bm{x}_j)\right) }{N} \mbox{E}_p \left[\left(\frac{I_i}{\pi_i}-1\right) \left(\frac{I_j}{\pi_j}-1\right) \Big| Q_n \right] \right\}\\
& = \mbox{E}_p\left\{  \underset{i\neq j: i, j \in U_{\sN{N}}}{\sum \sum} \frac{ \left(y_i- \tilde{h}_{\sN{N}}(\bm{x}_i)\right) }{N}\frac{ \left(y_j- \tilde{h}_{\sN{N}}(\bm{x}_j)\right) }{N}  \left[ \frac{(\pi_{ij} - \pi_i \pi_j) + d_{ij} - \pi_i d_j - \pi_j d_i}{\pi_i \pi_j} \right] \right\}\\
& \leq \left[\max_{i, j \in U_{N} :i \neq j} \left|\frac{ \pi_{  ij}}{\pi_{  i}\pi_{  j}} - 1 \right|  + \mbox{E}_p \left(\max_{i, j \in U_{N} :i \neq j} \left|\frac{ d_{  ij}}{\pi_{  i}\pi_{  j}} \right| \right) + 2 \mbox{E}_p \left(\max_{i \in U_N} |d_i|\right) \frac{1}{ \min_{i \in U_N} \pi_i} \right]  4M^2\\
& = O(n^{-1}) + O(n k(n)^{-2}) + O( n_{ }^{1/2}\gamma(n_{ })^{-1}k(n_{ })^{-1})\\
& = o(\gamma(n)^{-2})
\end{align*}
by conditions (\ref{cond1}) -- (\ref{cond5}) and rate condition (\ref{rate3}).

\noindent By rewriting the mean functions and applying the Cauchy Schwarz Inequality to $B_{\sN{N}}$, we obtain
\begin{align*}
B_{\sN{N}} &= \mbox{E}_p \left|\sum_{k=1}^q ( \tilde{\mu}_{\sN{N}k} - \tilde{\mu}_{nk} ) \sum_{i\in U_{\sN{N}}} \frac{ I(\bm{x}_i \in B_{nk})}{N}\left(\frac{I_i}{\pi_i}-1\right) \right| \\
&\leq \left\{ \sum_{k=1}^q \mbox{E}_p\left[( \tilde{\mu}_{\sN{N}k} - \tilde{\mu}_{nk} )^2 \right]\right\}^{1/2}  \left\{ \sum_{k=1}^q   \mbox{E}_p \left[\underset{i, j\in U_{\sN{N}}}{ \sum\sum} \frac{ I(\bm{x}_i \in B_{nk})}{N} \frac{ I(\bm{x}_j \in B_{nk})}{N}\left(\frac{I_i}{\pi_i}-1\right) \left(\frac{I_j}{\pi_j}-1\right) \right] \right\}^{1/2}\\
& \leq \left\{ \frac{n}{k(n)} O\left(\frac{ n_{ }^{1/2}\gamma(n_{ })}{k(n)} \right) \right\}^{1/2}  \left\{  \mbox{E}_p \left[\underset{i, j\in U_{\sN{N}}}{ \sum\sum} \frac{1}{N^2}\left(\frac{I_i}{\pi_i}-1\right) \left(\frac{I_j}{\pi_j}-1\right) \sum_{k=1}^q I(\bm{x}_i \in B_{nk})  I(\bm{x}_j \in B_{nk}) \right] \right\}^{1/2}\\
& \leq \left\{ O\left(\frac{ n_{ }^{3/2}\gamma(n_{ })}{k(n)^{2}} \right) \right\}^{1/2}   \left\{  \mbox{E}_p \left[\underset{i, j\in U_{\sN{N}}}{ \sum\sum} \frac{1}{N^2}\left(\frac{I_i}{\pi_i}-1\right) \left(\frac{I_j}{\pi_j}-1\right)  \right] \right\}^{1/2}\\
& \leq \left\{ O\left(\frac{ n_{ }^{3/2}\gamma(n_{ })}{k(n)^{2}} \right) \right\}^{1/2}   O(n^{-1/2})\\
& = O\left(\left[ \frac{1}{k(n)} \frac{ n_{ }^{1/2}\gamma(n_{ })}{k(n)}\right]^{1/2}   \right)\\
& = o(1)
\end{align*}
by Proposition 1 in \citet*{tot11} and by Conditions (\ref{pimin}) and (\ref{cond3}).

\end{proof}


 \bibliography{RegTreeEst}

\end{document}